\documentclass{article}%
\usepackage{amssymb}
\usepackage{amsmath}
\usepackage{amsfonts}
\usepackage{graphicx}%
\providecommand{\U}[1]{\protect\rule{.1in}{.1in}}
\newtheorem{theorem}{Theorem}[section]

\newtheorem{conjecture}[theorem]{Conjecture}
\newtheorem{corollary}[theorem]{Corollary}

\newtheorem{lemma}[theorem]{Lemma}

\newtheorem{problem}[theorem]{Problem}
\newtheorem{proposition}[theorem]{Proposition}
\newtheorem{remark}[theorem]{Remark}

\newenvironment{proof}[1][Proof]{\noindent\textbf{#1.} }{\ \rule{0.5em}{0.5em}}
\begin{document}

\title{Critical Independent Sets of a Graph}
\author{Vadim E. Levit\\Department of Computer Science and Mathematics\\Ariel University, Israel\\levitv@ariel.ac.il
\and Eugen Mandrescu\\Department of Computer Science\\Holon Institute of Technology, Israel\\eugen\_m@hit.ac.il}
\date{}
\maketitle

\begin{abstract}
Let $G$ be a simple graph with vertex set $V\left(  G\right)  $. A set
$S\subseteq V\left(  G\right)  $ is \textit{independent} if no two vertices
from $S$ are adjacent, and by $\mathrm{Ind}(G)$ we mean the family of all
independent sets of $G$.

The number $d\left(  X\right)  =$ $\left\vert X\right\vert -\left\vert
N(X)\right\vert $ is the \textit{difference} of $X\subseteq V\left(  G\right)
$, and a set $A\in\mathrm{Ind}(G)$ is \textit{critical} if $d(A)=\max
\{d\left(  I\right)  :I\in\mathrm{Ind}(G)\}$ \cite{Zhang1990}.

Let us recall the following definitions:%
\begin{align*}
\mathrm{core}\left(  G\right)   &  =%
{\displaystyle\bigcap}
\left\{  S:S\text{\textit{ is a maximum independent set}}\right\}  \text{
\cite{LevMan2002a},}\\
\mathrm{corona}\left(  G\right)   &  =%
{\displaystyle\bigcup}
\left\{  S:S\text{\textit{ is a maximum independent set}}\right\}  \text{
\cite{BorosGolLev},}\\
\mathrm{\ker}(G)  &  =%
{\displaystyle\bigcap}
\left\{  S:S\text{\textit{ is a critical independent set}}\right\}  \text{
\cite{LevMan2012a},}\\
\mathrm{diadem}(G)  &  =%
{\displaystyle\bigcup}
\left\{  S:S\text{\textit{ is a critical independent set}}\right\}  \text{.}%
\end{align*}

In this paper we present various structural properties of $\mathrm{\ker}(G)$,
in relation with $\mathrm{core}\left(  G\right)  $, $\mathrm{corona}\left(
G\right)  $, and $\mathrm{diadem}(G)$.

\textbf{Keywords:} independent set, critical set, ker, core, corona, diadem, matching

\end{abstract}

\section{Introduction}

Throughout this paper $G$ is a finite simple graph with vertex set $V(G)$ and
edge set $E(G)$. If $X\subseteq V\left(  G\right)  $, then $G[X]$ is the
subgraph of $G$ induced by $X$. By $G-W$ we mean either the subgraph
$G[V\left(  G\right)  -W]$, if $W\subseteq V(G)$, or the subgraph obtained by
deleting the edge set $W$, for $W\subseteq E(G)$. In either case, we use
$G-w$, whenever $W$ $=\{w\}$. If $A,B$ $\subseteq V\left(  G\right)  $, then
$(A,B)$ stands for the set $\{ab:a\in A,b\in B,ab\in E\left(  G\right)  \}$.

The \textit{neighborhood} $N(v)$ of a vertex $v\in V\left(  G\right)  $ is the
set $\{w:w\in V\left(  G\right)  $ \textit{and} $vw\in E\left(  G\right)  \}$,
while the \textit{closed neighborhood} $N[v]$\ of $v\in V\left(  G\right)  $
is the set $N(v)\cup\{v\}$; in order to avoid ambiguity, we use also
$N_{G}(v)$ instead of $N(v)$.

The \textit{neighborhood} $N(A)$ of $A\subseteq V\left(  G\right)  $ is
$\{v\in V\left(  G\right)  :N(v)\cap A\neq\emptyset\}$, and $N[A]=N(A)\cup A$.
We may also use $N_{G}(A)$ and $N_{G}\left[  A\right]  $, when referring to
neighborhoods in a graph $G$.

A set $S\subseteq V(G)$ is \textit{independent} if no two vertices from $S$
are adjacent, and by $\mathrm{Ind}(G)$ we mean the family of all the
independent sets of $G$. An independent set of maximum size is a
\textit{maximum independent set} of $G$, and the \textit{independence number
}$\alpha(G)$ of $G$ is $\max\{\left\vert S\right\vert :S\in\mathrm{Ind}(G)\}$.
Let $\Omega(G)$ denote the family of all maximum independent sets, and let%
\begin{align*}
\mathrm{core}(G)  &  =%
{\displaystyle\bigcap}
\{S:S\in\Omega(G)\}\text{ \cite{LevMan2002a}, and}\\
\mathrm{corona}(G)  &  =\cup\{S:S\in\Omega(G)\}\text{ \cite{BorosGolLev}.}%
\end{align*}
Clearly, $N\left(  \mathrm{core}(G)\right)  \subseteq$ $V\left(  G\right)
-\mathrm{corona}(G)$, and there are graphs with $N\left(  \mathrm{core}%
(G)\right)  $ $\neq$ $V\left(  G\right)  -\mathrm{corona}(G)$ (for an example,
see Figure \ref{fig101}). The problem of whether $\mathrm{core}(G)\neq
\emptyset$ is \textbf{NP}-hard \cite{BorosGolLev}.

\begin{figure}[h]
\setlength{\unitlength}{1.0cm} \begin{picture}(5,1.75)\thicklines
\multiput(5,0.5)(1,0){5}{\circle*{0.29}}
\multiput(5,1.5)(1,0){5}{\circle*{0.29}}
\put(5,0.5){\line(1,0){4}}
\put(5,1.5){\line(1,-1){1}}
\put(6,1.5){\line(1,-1){1}}
\put(6,1.5){\line(1,0){1}}
\put(7,0.5){\line(0,1){1}}
\put(7,0.5){\line(1,1){1}}
\put(8,1.5){\line(1,0){1}}
\put(9,0.5){\line(0,1){1}}
\put(4.65,0.5){\makebox(0,0){$a$}}
\put(4.65,1.5){\makebox(0,0){$b$}}
\put(6,0.2){\makebox(0,0){$c$}}
\put(7,0.2){\makebox(0,0){$d$}}
\put(8,0.2){\makebox(0,0){$e$}}
\put(9.3,0.5){\makebox(0,0){$f$}}
\put(5.65,1.5){\makebox(0,0){$x$}}
\put(7.3,1.5){\makebox(0,0){$y$}}
\put(8,1.2){\makebox(0,0){$u$}}
\put(9.3,1.5){\makebox(0,0){$v$}}
\put(3.3,1){\makebox(0,0){$G$}}
\end{picture}\caption{\textrm{core}$(G)=\{a,b\}$ and $V(G)-$ \textrm{corona}%
$(G)=N\left(  \mathrm{core}(G)\right)  \cup\{d\}=\{c,d\}$.}%
\label{fig101}%
\end{figure}
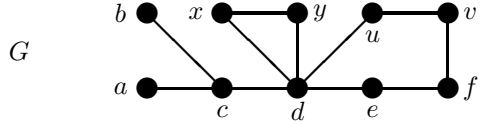

A \textit{matching} is a set $M$ of pairwise non-incident edges of $G$. If
$A\subseteq V(G)$, then $M\left(  A\right)  $ is the set of all the vertices
matched by $M$ with vertices belonging to $A$. A matching of maximum
cardinality, denoted $\mu(G)$, is a \textit{maximum matching}.

For $X\subseteq V(G)$, the number $\left\vert X\right\vert -\left\vert
N(X)\right\vert $ is the \textit{difference} of $X$, denoted $d(X)$. The
\textit{critical difference} $d(G)$ is $\max\{d(X):X\subseteq V(G)\}$. The
number $\max\{d(I):I\in\mathrm{Ind}(G)\}$ is the \textit{critical independence
difference} of $G$, denoted $id(G)$. Clearly, $d(G)\geq id(G)$. It was shown
in \cite{Zhang1990} that $d(G)$ $=id(G)$ holds for every graph $G$. If $A$ is
an independent set in $G$ with $d\left(  X\right)  =id(G)$, then $A$ is a
\textit{critical independent set} \cite{Zhang1990}. All pendant vertices not
belonging to $K_{2}$ components are included in every inclusion maximal
critical independent set.

For example, let $X=\{v_{1},v_{2},v_{3},v_{4}\}$ and $I=\{v_{1},v_{2}%
,v_{3},v_{6},v_{7}\}$ in the graph $G$ of Figure \ref{fig511}. Note that $X$
is a critical set, since $N(X)=\{v_{3},v_{4},v_{5}\}$ and $d(X)=1=d(G)$, while
$I$ is a critical independent set, because $d(I)=1=id(G)$. Other critical sets
are $\{v_{1},v_{2}\}$, $\{v_{1},v_{2},v_{3}\}$, $\{v_{1},v_{2},v_{3}%
,v_{4},v_{6},v_{7}\}$. \begin{figure}[h]
\setlength{\unitlength}{1cm}\begin{picture}(5,1.9)\thicklines
\multiput(6,0.5)(1,0){6}{\circle*{0.29}}
\multiput(5,1.5)(1,0){4}{\circle*{0.29}}
\multiput(4,0.5)(0,1){2}{\circle*{0.29}}
\put(10,1.5){\circle*{0.29}}
\put(4,0.5){\line(1,0){7}}
\put(4,1.5){\line(2,-1){2}}
\put(5,1.5){\line(1,-1){1}}
\put(5,1.5){\line(1,0){1}}
\put(6,0.5){\line(0,1){1}}
\put(6,0.5){\line(1,1){1}}
\put(7,1.5){\line(1,0){1}}
\put(8,0.5){\line(0,1){1}}
\put(10,0.5){\line(0,1){1}}
\put(10,1.5){\line(1,-1){1}}
\put(4,0.1){\makebox(0,0){$v_{1}$}}
\put(3.65,1.5){\makebox(0,0){$v_{2}$}}
\put(4.65,1.5){\makebox(0,0){$v_{3}$}}
\put(6.35,1.5){\makebox(0,0){$v_{4}$}}
\put(6,0.1){\makebox(0,0){$v_{5}$}}
\put(7,0.1){\makebox(0,0){$v_{6}$}}
\put(7,1.15){\makebox(0,0){$v_{7}$}}
\put(8,0.1){\makebox(0,0){$v_{9}$}}
\put(8.35,1.5){\makebox(0,0){$v_{8}$}}
\put(9.65,1.5){\makebox(0,0){$v_{11}$}}
\put(9,0.1){\makebox(0,0){$v_{10}$}}
\put(10,0.1){\makebox(0,0){$v_{12}$}}
\put(11,0.1){\makebox(0,0){$v_{13}$}}
\put(2.5,1){\makebox(0,0){$G$}}
\end{picture}\caption{\textrm{core}$(G)=\{v_{1},v_{2},v_{6},v_{10}\}$ is a
critical set.}%
\label{fig511}%
\end{figure}
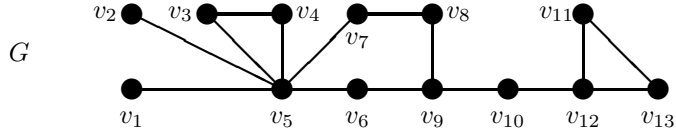

It is known that finding a maximum independent set is an \textbf{NP}-hard
problem \cite{GaryJohnson79}. Zhang proved that a critical independent set can
be find in polynomial time \cite{Zhang1990}. A simpler algorithm, reducing the
critical independent set problem to computing a maximum independent set in a
bipartite graph is given in \cite{Ageev}.

\begin{theorem}
\label{th3}\cite{ButTruk2007} Each critical independent set can be enlarged to
a maximum independent set.
\end{theorem}

Theorem \ref{th3} led to an efficient way of approximating $\alpha(G)$
\cite{Truchanov2008}. Moreover, it has been shown that a critical independent
set of maximum cardinality can be computed in polynomial time
\cite{Larson2007}. Recently, a parallel algorithm computing the critical
independence number was developed \cite{DeLaVinaLarson2013}.

Recall that if $\alpha(G)+\mu(G)=\left\vert V(G)\right\vert $, then $G$ is a
\textit{K\"{o}nig-Egerv\'{a}ry graph} \cite{Deming1979,Sterboul1979}. As a
well-known example, each bipartite graph is a K\"{o}nig-Egerv\'{a}ry graph as well.

\begin{theorem}
\label{Th5} \cite{LevMan2003} If $G$ is a K\"{o}nig-Egerv\'{a}ry graph, $M$ is
a maximum matching of $G$, and $S\in\Omega\left(  G\right)  $, then:

\emph{(i)} $M$ matches $V\left(  G\right)  -S$ into $S$, and $N(\mathrm{core}%
(G))$ into $\mathrm{core}(G)$;

\emph{(ii)} $N\left(  \mathrm{core}(G)\right)  =\cap\left\{  V(G)-S:S\in
(G)\right\}  $, i.e., $N\left(  \mathrm{core}(G\right)  )=V\left(  G\right)
-\mathrm{corona}(G)$.
\end{theorem}

The \textit{deficiency} $def(G)$ is the number of non-saturated vertices
relative to a maximum matching, i.e., $def(G)=\left\vert V\left(  G\right)
\right\vert -2\mu(G)$ \cite{LovPlum1986}. A proof of a conjecture of
Graffiti.pc \cite{DeLaVina} yields a new characterization of
K\"{o}nig-Egerv\'{a}ry graphs: these are exactly the graphs, where there
exists a critical maximum independent set \cite{Larson2011}. In
\cite{LevMan2012b} it is proved the following.

\begin{theorem}
\label{th8}\cite{LevMan2012b} For a K\"{o}nig-Egerv\'{a}ry graph $G$ the
following equalities hold
\[
d(G)=\left\vert \mathrm{core}(G)\right\vert -\left\vert N(\mathrm{core}%
(G))\right\vert =\alpha(G)-\mu(G)=def(G).
\]

\end{theorem}

Using this finding, we have strengthened the characterization from
\cite{Larson2011}.

\begin{theorem}
\label{th5}\cite{LevMan2012b} $G$ is a K\"{o}nig-Egerv\'{a}ry graph if and
only if each of its maximum independent sets is critical.
\end{theorem}

For a graph $G$, let denote
\begin{align*}
\mathrm{\ker}(G)  &  =%
{\displaystyle\bigcap}
\left\{  S:S\text{ \textit{is a critical independent set}}\right\}  \text{
\cite{LevMan2012a}, and}\\
\mathrm{diadem}(G)  &  =\bigcup\left\{  S:S\text{ \textit{is a critical
independent set}}\right\}  \text{.}%
\end{align*}

In this paper we present several properties of $\mathrm{\ker}(G)$, in relation
with $\mathrm{core}(G)$, $\mathrm{corona}(G)$, and $\mathrm{diadem}(G)$.

\section{Preliminaries}

Let $G$ be the graph from Figure \ref{fig511}; the sets $X=\left\{
v_{1},v_{2},v_{3}\right\}  $, $Y=\left\{  v_{1},v_{2},v_{4}\right\}  $\ are
critical independent, and the sets $X\cap Y$, $X\cup Y$ are also critical, but
only $X\cap Y$ is also independent.\ In addition, one can easily see that
$\mathrm{\ker}(G)$ is a minimal critical independent set of $G$. These
properties of critical sets and $\mathrm{\ker}(G)$ are true even in general.

\begin{theorem}
\label{th4}\cite{LevMan2012a} For a graph $G$, the following assertions are true:

\emph{(i)} the function $d$ is supermodular, i.e., $d(A\cup B)+d(A\cap B)\geq
d(A)+d(B)$ for every $A,B\subseteq V(G)$;

\emph{(ii)} if $A$ and $B$ are critical in $G$, then $A\cup B$ and $A\cap B$
are critical as well;

\emph{(iii)} $G$ has a unique minimal independent critical set, namely,
$\mathrm{\ker}(G)$.
\end{theorem}

As a consequence, we have the following.

\begin{corollary}
For every graph $G$, $\mathrm{diadem}(G)$ is a critical set.
\end{corollary}

For instance, the graph $G$ from Figure \ref{fig511} has $\mathrm{diadem}%
(G)=\left\{  v_{1},v_{2},v_{3},v_{4},v_{6},v_{7},v_{10}\right\}  $, which is
critical, but not independent.

\begin{figure}[h]
\setlength{\unitlength}{1cm}\begin{picture}(5,1.8)\thicklines
\multiput(1,0.5)(1,0){6}{\circle*{0.29}}
\multiput(1,1.5)(1,0){5}{\circle*{0.29}}
\put(1,0.5){\line(1,0){5}}
\put(1,1.5){\line(1,-1){1}}
\put(2,0.5){\line(0,1){1}}
\put(3,1.5){\line(1,-1){1}}
\put(3,1.5){\line(1,0){1}}
\put(4,0.5){\line(0,1){1}}
\put(5,1.5){\line(1,-1){1}}
\put(5,0.5){\line(0,1){1}}
\put(0.7,1.5){\makebox(0,0){$a$}}
\put(0.7,0.5){\makebox(0,0){$b$}}
\put(2.3,1.5){\makebox(0,0){$c$}}
\put(3,0.85){\makebox(0,0){$d$}}
\put(3,0){\makebox(0,0){$G_{1}$}}
\multiput(7,0.5)(1,0){7}{\circle*{0.29}}
\multiput(7,1.5)(1,0){6}{\circle*{0.29}}
\put(7,0.5){\line(1,0){6}}
\put(7,1.5){\line(1,-1){1}}
\put(8,0.5){\line(0,1){1}}
\put(9,0.5){\line(0,1){1}}
\put(9,0.5){\line(1,1){1}}
\put(9,1.5){\line(1,0){1}}
\put(11,0.5){\line(0,1){1}}
\put(12,1.5){\line(1,-1){1}}
\put(11,1.5){\line(1,0){1}}
\put(6.7,1.5){\makebox(0,0){$x$}}
\put(6.7,0.5){\makebox(0,0){$y$}}
\put(8.3,1.5){\makebox(0,0){$z$}}
\put(10,0.8){\makebox(0,0){$w$}}
\put(10,0){\makebox(0,0){$G_{2}$}}
\end{picture}\caption{Both $G_{1}$ and $G_{2}$\ are not K\"{o}nig-Egerv\'{a}ry
graphs.}%
\label{fig22}%
\end{figure}
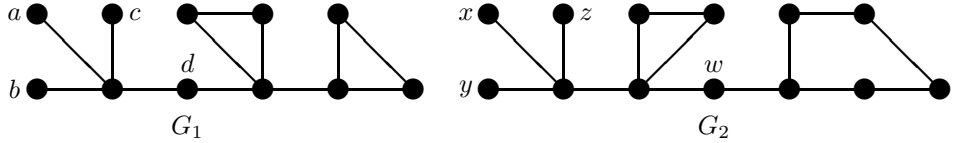

The graph $G$ from Figure \ref{fig101} has $d\left(  G\right)  =1$ and
$d\left(  \mathrm{corona}(G)\right)  =0$, which means that $\mathrm{corona}%
(G)$ is not a critical set. Notice that $G$ is not a K\"{o}nig-Egerv\'{a}ry
graph. Combining Theorems \ref{th5} and \ref{th4}\emph{(ii)}, we deduce the following.

\begin{corollary}
\label{cor2}If $G$ is a K\"{o}nig-Egerv\'{a}ry graph, then both $\mathrm{core}%
(G)$ and $\mathrm{corona}(G)$ are critical sets.
\end{corollary}

Let consider the graphs $G_{1}$ and $G_{2}$ from Figure \ref{fig22}:
$\mathrm{core}(G_{1})=\left\{  a,b,c,d\right\}  $ and it is a critical set,
while $\mathrm{core}(G_{2})=\left\{  x,y,z,w\right\}  $ and it is not critical.

\begin{theorem}
If $\mathrm{core}(G)$ is a critical set, then
\[
\mathrm{core}(G)\subseteq%
{\displaystyle\bigcap}
\left\{  A:A\text{ \textit{is an inclusion maximal critical independent set}%
}\right\}  .
\]

\end{theorem}

\begin{proof}
Let $A$ be an arbitrary inclusion maximal critical independent set. According
to Theorem \ref{th3}, there is some $S\in\Omega\left(  G\right)  $, such that
$A\subseteq S$. Since $\mathrm{core}(G)\subseteq S$, it follows that
$A\cup\mathrm{core}(G)\subseteq S$, and hence $A\cup\mathrm{core}(G)$ is
independent. By Theorem \ref{th4}, we get that $A\cup\mathrm{core}(G)$ is a
critical independent set. Since $A\subseteq A\cup\mathrm{core}(G)$ and $A$ is
an inclusion maximal critical independent set, it follows that $\mathrm{core}%
(G)\subseteq A$, for every such set $A$, and this completes the proof.
\end{proof}

\begin{remark}
By Theorem \ref{th3} the following inclusion holds for every graph $G$.
\[
\mathrm{corona}(G)\supseteq%
{\displaystyle\bigcup}
\left\{  A:A\text{ \textit{is an inclusion maximal critical independent set}%
}\right\}  .
\]

\end{remark}

\section{Structural properties of $\mathrm{\ker}\left(  G\right)  $}

Deleting a vertex from a graph may change its critical difference. For
instance, $d\left(  G-v_{1}\right)  =d\left(  G\right)  -1$, $d\left(
G-v_{13}\right)  =d\left(  G\right)  $, while $d\left(  G-v_{3}\right)
=d\left(  G\right)  +1$, where $G$ is the graph of Figure \ref{fig511}.

\begin{proposition}
\cite{LevMan2013a} For a vertex $v$ in a graph $G$, the following assertions hold:

\emph{(i) }$d\left(  G-v\right)  =d\left(  G\right)  -1$ if and only if
$v\in\mathrm{\ker}(G)$;

\emph{(ii)} if $v\in\mathrm{\ker}(G)$, then $\mathrm{\ker}(G-v)\subseteq
\mathrm{\ker}(G)-\left\{  v\right\}  $.
\end{proposition}

Note that $\mathrm{\ker}(G-v)$ may differ from $\mathrm{\ker}(G)-\left\{
v\right\}  $. For example, $\mathrm{\ker}(K_{3,2})$ is equal to the partite
set of size $3$, but $\mathrm{\ker}(K_{3,2}-v)=\emptyset$ whenever $v$ is in
that set. Also, if $G=C_{4}$, then $\mathrm{\ker}(G)-\left\{  v\right\}
=\emptyset-\left\{  v\right\}  =\emptyset$, while $\mathrm{\ker}%
(G-v)=N_{G}(v)$ for every $v\in V\left(  G\right)  $.

\begin{theorem}
\label{th2}\cite{Larson2007} There is a matching from $N(S)$ into $S$ for
every critical independent set $S$.
\end{theorem}

In the graph $G$ of Figure \ref{fig511}, let $S=\{v_{1},v_{2},v_{3}\}$. By
Theorem \ref{th2}, there is a matching from $N\left(  S\right)  $ into
$S=\left\{  v_{1},v_{2},v_{3}\right\}  $, for instance, $M=\left\{  v_{2}%
v_{5},v_{3}v_{4}\right\}  $, since $S$ is critical independent. On the other
hand, there is no matching from $N\left(  S\right)  $ into $S-v_{3}$.

\begin{theorem}
\cite{LevMan2013a}\label{th9} For a critical independent set $A$ in a graph
$G$, the following statements are equivalent:

\emph{(i) }$A=\mathrm{\ker}(G)$;

\emph{(ii)} there is no set $B\subseteq$ $N\left(  A\right)  ,B\neq\emptyset$
such that $\left\vert N\left(  B\right)  \cap A\right\vert =\left\vert
B\right\vert $;

\emph{(iii) }for each $v\in A$ there exists a matching from $N\left(
A\right)  $ into $A-v$.
\end{theorem}

The graphs $G_{1}$ and $G_{2}$ in Figure \ref{fig333} satisfy $\mathrm{\ker
}(G_{1})=\mathrm{core}(G_{1})$, $\mathrm{\ker}(G_{2})=\left\{  x,y,z\right\}
\subset\mathrm{core}(G_{2})$, and both $\mathrm{core}(G_{1})$ and
$\mathrm{core}(G_{2})$\ are critical sets of maximum size. The graph $G_{3}$
in Figure \ref{fig333} has $\mathrm{\ker}(G_{3})=\{u,v\}$, the set
$\{t,u,v\}$\ as a critical independent set of maximum size, while
$\mathrm{core}(G_{3})=\left\{  t,u,v,w\right\}  $\ is not a critical set.

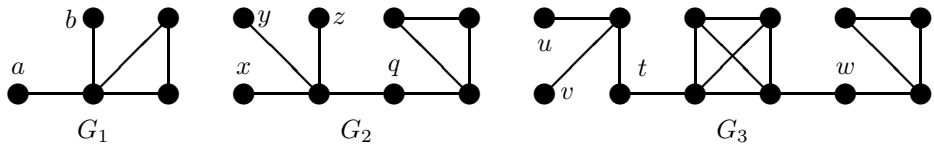
\begin{figure}[h]
\setlength{\unitlength}{1.0cm} \begin{picture}(5,1.9)\thicklines
\multiput(1,0.5)(1,0){3}{\circle*{0.29}}
\multiput(2,1.5)(1,0){2}{\circle*{0.29}}
\put(1,0.5){\line(1,0){2}}
\put(2,0.5){\line(0,1){1}}
\put(2,0.5){\line(1,1){1}}
\put(3,0.5){\line(0,1){1}}
\put(1,0.84){\makebox(0,0){$a$}}
\put(1.7,1.5){\makebox(0,0){$b$}}
\put(2,0){\makebox(0,0){$G_{1}$}}
\multiput(4,0.5)(1,0){4}{\circle*{0.29}}
\put(4,0.5){\line(1,0){3}}
\multiput(4,1.5)(1,0){4}{\circle*{0.29}}
\put(4,1.5){\line(1,-1){1}}
\put(5,0.5){\line(0,1){1}}
\put(7,0.5){\line(0,1){1}}
\put(7,0.5){\line(-1,1){1}}
\put(6,1.5){\line(1,0){1}}
\put(6,0.84){\makebox(0,0){$q$}}
\put(4,0.84){\makebox(0,0){$x$}}
\put(4.27,1.5){\makebox(0,0){$y$}}
\put(5.27,1.5){\makebox(0,0){$z$}}
\put(5.5,0){\makebox(0,0){$G_{2}$}}
\multiput(8,0.5)(1,0){6}{\circle*{0.29}}
\multiput(8,1.5)(1,0){6}{\circle*{0.29}}
\put(8,0.5){\line(1,1){1}}
\put(8,1.5){\line(1,0){1}}
\put(9,0.5){\line(0,1){1}}
\put(9,0.5){\line(1,0){4}}
\put(10,0.5){\line(0,1){1}}
\put(10,0.5){\line(1,1){1}}
\put(10,1.5){\line(1,-1){1}}
\put(10,1.5){\line(1,0){1}}
\put(11,0.5){\line(0,1){1}}
\put(12,1.5){\line(1,-1){1}}
\put(12,1.5){\line(1,0){1}}
\put(13,0.5){\line(0,1){1}}
\put(8.3,0.5){\makebox(0,0){$v$}}
\put(8,1.15){\makebox(0,0){$u$}}
\put(9.3,0.84){\makebox(0,0){$t$}}
\put(12,0.84){\makebox(0,0){$w$}}
\put(10.5,0){\makebox(0,0){$G_{3}$}}
\end{picture}\caption{$\mathrm{core}(G_{1})=\left\{  a,b\right\}  $,
$\mathrm{core}(G_{2})=\left\{  q,x,y,z\right\}  $, $\mathrm{core}%
(G_{3})=\left\{  t,u,v,w\right\}  $.}%
\label{fig333}%
\end{figure}

An independent set $S$ is \textit{inclusion minimal with} $d\left(  S\right)
>0$ if no proper subset of $S$ has positive difference. For example, in Figure
\ref{fig333} one can see that $\mathrm{\ker}(G_{1})$ is an inclusion minimal
independent set with positive difference, while for the graph $G_{2}$ the sets
$\{x,y\},\{x,z\},\{y,z\}$ are inclusion minimal independent with positive
difference, and $\mathrm{\ker}(G_{2})=\{x,y\}\cup\{x,z\}\cup\{y,z\}$.

\begin{theorem}
\cite{LevMan2013a}\label{th1} If $\mathrm{\ker}(G)\neq\emptyset$, then
\begin{align*}
\mathrm{\ker}(G)  &  =%
{\displaystyle\bigcup}
\left\{  S_{0}:S_{0}\text{ is an inclusion minimal independent set with
}d\left(  S_{0}\right)  =1\right\} \\
&  =%
{\displaystyle\bigcup}
\left\{  S_{0}:S_{0}\text{ is an inclusion minimal independent set with
}d\left(  S_{0}\right)  >0\right\}  .
\end{align*}

\end{theorem}

In a graph $G$, the union of all minimum cardinality independent sets $S$ with
$d\left(  S\right)  >0$ may be a proper subset of $\mathrm{\ker}\left(
G\right)  $. For example, consider the graph $G$ in Figure \ref{fig177}, where
$\left\{  x,y\right\}  \subset\mathrm{\ker}\left(  G\right)  =\left\{
x,y,u,v,w\right\}  $.

\begin{figure}[h]
\setlength{\unitlength}{1cm}\begin{picture}(5,1.3)\thicklines
\multiput(4,0)(1,0){3}{\circle*{0.29}}
\multiput(3,1)(1,0){5}{\circle*{0.29}}
\put(4,0){\line(1,0){2}}
\put(4,0){\line(0,1){1}}
\put(3,1){\line(1,-1){1}}
\put(5,0){\line(0,1){1}}
\put(5,0){\line(2,1){2}}
\put(6,0){\line(0,1){1}}
\put(6,0){\line(1,1){1}}
\put(2.7,1){\makebox(0,0){$x$}}
\put(3.7,1){\makebox(0,0){$y$}}
\put(4.7,1){\makebox(0,0){$u$}}
\put(5.7,1){\makebox(0,0){$v$}}
\put(7.3,1){\makebox(0,0){$w$}}
\put(2,0.5){\makebox(0,0){$G$}}
\multiput(9,0)(1,0){3}{\circle*{0.29}}
\multiput(9,1)(2,0){2}{\circle*{0.29}}
\put(9,0){\line(1,0){2}}
\put(9,1){\line(1,-1){1}}
\put(10,0){\line(1,1){1}}
\put(8.65,0){\makebox(0,0){$v_{1}$}}
\put(8.65,1){\makebox(0,0){$v_{2}$}}
\put(11.35,1){\makebox(0,0){$v_{3}$}}
\put(11.35,0){\makebox(0,0){$v_{4}$}}
\put(8,0.5){\makebox(0,0){$H$}}
\end{picture}\caption{Both $S_{1}=\{x,y\}$ and $S_{2}=\{u,v,w\}$ are inclusion
minimal independent sets satisfying $d\left(  S\right)  >0$.}%
\label{fig177}%
\end{figure}
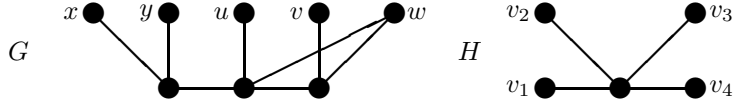

Actually, all inclusion minimal independent sets $S$ with $d(S)>0$ are of the
same difference.

\begin{proposition}
\cite{LevMan2013a}\label{prop3} If $S_{0}$ is an inclusion minimal independent
set with $d\left(  S_{0}\right)  >0$, then $d\left(  S_{0}\right)  =1$. In
other words,%
\begin{gather*}
\left\{  S_{0}:S_{0}\text{ is an inclusion minimal independent set with
}d\left(  S_{0}\right)  >0\right\}  =\\
=\left\{  S_{0}:S_{0}\text{ is an inclusion minimal independent set with
}d\left(  S_{0}\right)  =1\right\}  .
\end{gather*}

\end{proposition}

The converse of Proposition \ref{prop3} is not true. For instance, $S=\left\{
x,y,u\right\}  $ is independent in the graph $G$ of Figure \ref{fig177}\ and
$d\left(  S\right)  =1$, but $S$ is not minimal with this property.

\begin{proposition}
\cite{LevMan2013a} $\min\left\{  \left\vert S_{0}\right\vert :d\left(
S_{0}\right)  >0,S_{0}\in\mathrm{Ind}(G)\right\}  \leq\left\vert \mathrm{\ker
}\left(  G\right)  \right\vert -d\left(  G\right)  +1$ is true for every graph
$G$.
\end{proposition}

\section{Relationships between $\mathrm{\ker}\left(  G\right)  $ and
$\mathrm{core}(G)$}

Let us consider again the graph $G_{2}$ from Figure \ref{fig22}:
$\mathrm{core}(G_{2})=\left\{  x,y,z,w\right\}  $ and it is not critical, but
$\mathrm{\ker}\left(  G_{2}\right)  =\left\{  x,y,z\right\}  \subseteq
\mathrm{core}(G_{2})$. Clearly, the same inclusion holds for $G_{1}$, whose
$\mathrm{core}(G_{1})$ is a critical set.

\begin{theorem}
\label{th6}\cite{LevMan2012a} For every graph $G$, $\mathrm{\ker}%
(G)\subseteq\mathrm{core}(G)$.
\end{theorem}

Let $I_{c}$ be a maximum critical independent set of $G$, and $X=I_{c}\cup
N(I_{c})$. In \cite{Short} it is proved that $\mathrm{core}(G\left[  X\right]
)\subseteq\mathrm{core}(G)$. Moreover, in \cite{LevMan2012a}, we showed that
the chain of relationships $\mathrm{\ker}(G)=\mathrm{\ker}(G\left[  X\right]
)\subseteq\mathrm{core}(G\left[  X\right]  )\subseteq\mathrm{core}(G)$ holds
for every graph $G$. Theorem \ref{th6} allows an alternative proof of the
following inequality due to Lorentzen.

\begin{corollary}
\label{cor1}\cite{Lorentzen1966,Schrijver2003,LevMan2012a} The inequality
$d\left(  G\right)  \geq\alpha\left(  G\right)  -\mu\left(  G\right)  $ holds
for every graph.
\end{corollary}

Following Ore \cite{Ore55}, \cite{Ore62}, the number $\delta(X)=d\left(
X\right)  =\left\vert X\right\vert -\left\vert N\left(  X\right)  \right\vert
$ is the \textit{deficiency} of $X$, where $X\subseteq A$ or $X\subseteq B$
and $G=(A,B,E)$ is a bipartite graph. Let
\[
\delta_{0}(A)=\max\{\delta(X):X\subseteq A\},\quad\delta_{0}(B)=\max
\{\delta(Y):Y\subseteq B\}.
\]

A subset $X\subseteq A$ having $\delta(X)=\delta_{0}(A)$ is $A$%
-\textit{critical}, while $Y\subseteq B$ having $\delta(B)=\delta_{0}(B)$ is
$B$-\textit{critical}. For a bipartite graph $G=\left(  A,B,E\right)  $ let us
denote $\mathrm{\ker}_{A}(G)=\cap\left\{  S:S\text{ is }%
A\text{-\textit{critical}}\right\}  $ and \textrm{diadem}$_{A}(G)=\cup\left\{
S:S\text{ is }A\text{-\textit{critical}}\right\}  $. Similarly, $\mathrm{\ker
}_{B}(G)=\cap\left\{  S:S\text{ is }B\text{-\textit{critical}}\right\}  $ and
\textrm{diadem}$_{B}(G)=\cup\left\{  S:S\text{ is }B\text{-\textit{critical}%
}\right\}  $.

It is convenient to define $d\left(  \emptyset\right)  =\delta(\emptyset)=0$.

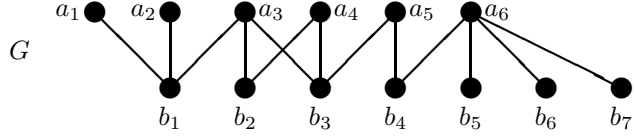
\begin{figure}[h]
\setlength{\unitlength}{1cm}\begin{picture}(5,1.9)\thicklines
\multiput(5,0.5)(1,0){7}{\circle*{0.29}}
\multiput(4,1.5)(1,0){6}{\circle*{0.29}}
\put(5,0.5){\line(-1,1){1}}
\put(5,0.5){\line(0,1){1}}
\put(5,0.5){\line(1,1){1}}
\put(6,0.5){\line(0,1){1}}
\put(6,0.5){\line(1,1){1}}
\put(7,0.5){\line(-1,1){1}}
\put(7,0.5){\line(0,1){1}}
\put(7,0.5){\line(1,1){1}}
\put(8,0.5){\line(0,1){1}}
\put(8,0.5){\line(1,1){1}}
\put(9,0.5){\line(0,1){1}}
\put(9,1.5){\line(1,-1){1}}
\put(9,1.5){\line(2,-1){2}}
\put(3.65,1.5){\makebox(0,0){$a_{1}$}}
\put(4.65,1.5){\makebox(0,0){$a_{2}$}}
\put(6.35,1.5){\makebox(0,0){$a_{3}$}}
\put(7.35,1.5){\makebox(0,0){$a_{4}$}}
\put(8.35,1.5){\makebox(0,0){$a_{5}$}}
\put(9.35,1.5){\makebox(0,0){$a_{6}$}}
\put(5,0.1){\makebox(0,0){$b_{1}$}}
\put(6,0.1){\makebox(0,0){$b_{2}$}}
\put(7,0.1){\makebox(0,0){$b_{3}$}}
\put(8,0.1){\makebox(0,0){$b_{4}$}}
\put(9,0.1){\makebox(0,0){$b_{5}$}}
\put(10,0.1){\makebox(0,0){$b_{6}$}}
\put(11,0.1){\makebox(0,0){$b_{7}$}}
\put(3,1){\makebox(0,0){$G$}}
\end{picture}\caption{$G$ is a bipartite graph without perfect matchings.}%
\label{fig233}%
\end{figure}For instance, the graph $G=(A,B,E)$ from Figure \ref{fig233} has:
$X=\left\{  a_{1},a_{2},a_{3},a_{4}\right\}  $ as an $A$-critical set,
$\mathrm{\ker}_{A}(G)=\left\{  a_{1},a_{2}\right\}  $, \textrm{diadem}%
$_{A}(G)=\left\{  a_{i}:i=1,...,5\right\}  $ and $\delta_{0}(A)=1$, while
$Y=\left\{  b_{i}:i=4,5,6,7\right\}  $ is a $B$-critical set, $\mathrm{\ker
}_{B}(G)=\left\{  b_{4},b_{5},b_{6}\right\}  $, \textrm{diadem}$_{B}%
(G)=\left\{  b_{i}:i=2,...,7\right\}  $ and $\delta_{0}(B)=2$.

As expected, there is a close relationship between critical independent sets
and $A$-critical\textit{ }or\textit{ }$B$-critical sets.

\begin{theorem}
\cite{LevMan2011b}\label{Th4} Let $G=\left(  A,B,E\right)  $ be a bipartite
graph. Then the following assertions are true:

\emph{(i) }$d(G)=\delta_{0}(A)+\delta_{0}(B)$;

\emph{(ii) }$\alpha\left(  G\right)  =\left\vert A\right\vert +\delta
_{0}(B)=\left\vert B\right\vert +\delta_{0}(A)=\mu\left(  G\right)
+\delta_{0}(A)+\delta_{0}(B)=\mu\left(  G\right)  +d\left(  G\right)  $;

\emph{(iii)} if $X$ is an $A$-\textit{critical set and }$Y$\textit{ is a }%
$B$-\textit{critical set, then }$X\cup Y$ is a critical set;

\emph{(iv)} if $Z$ is a critical independent set, then $Z\cap A$ is an
$A$-\textit{critical set }and $Z\cap B$ is \textit{a }$B$-\textit{critical
set;}

\emph{(v)} if $X$ is either an $A$-\textit{critical set or a }$B$%
-\textit{critical set, then there is a matching from }$N\left(  X\right)  $
into $X$.
\end{theorem}

The following lemma will be used further to give an alternative proof for the
assertion that $\mathrm{\ker}(G)=\mathrm{core}(G)$ holds for every bipartite
graph $G$.

\begin{lemma}
\label{lem1}If $G=\left(  A,B,E\right)  $ is a bipartite graph with a perfect
matching, say $M$, $S\in\Omega\left(  G\right)  $, $X\in$ $\mathrm{Ind}(G)$,
$X\subseteq V\left(  G\right)  -S$, and $G\left[  X\cup M\left(  X\right)
\right]  $ is connected, then
\[
X^{1}=X\cup M\left(  \left(  N\left(  X\right)  \cap S\right)  -M\left(
X\right)  \right)
\]
is an independent set, and $G\left[  X^{1}\cup M\left(  X^{1}\right)  \right]
$ is connected.
\end{lemma}

\begin{proof}
Let us show that the set $M\left(  \left(  N\left(  X\right)  \cap S\right)
-M\left(  X\right)  \right)  $ is independent. Suppose, to the contrary, that
there exist $v_{1},v_{2}\in M\left(  \left(  N\left(  X\right)  \cap S\right)
-M\left(  X\right)  \right)  $ such that $v_{1}v_{2}\in E\left(  G\right)  $.
Hence $M\left(  v_{1}\right)  ,M\left(  v_{2}\right)  \in\left(  N\left(
X\right)  \cap S\right)  -M\left(  X\right)  $.

If $M\left(  v_{1}\right)  $ and $M\left(  v_{2}\right)  $ have a common
neighbor $w\in X$, then $\left\{  v_{1},v_{2},M\left(  v_{2}\right)
,w,M\left(  v_{1}\right)  \right\}  $ spans $C_{5}$, which is forbidden for
bipartite graphs.

Otherwise, let $w_{1},w_{2}\in X$ be neighbors of $M\left(  v_{1}\right)  $
and $M\left(  v_{2}\right)  $, respectively. Since $G\left[  X\cup M\left(
X\right)  \right]  $ is connected, there is a path with even number of edges
connecting $w_{1}$ and $w_{2}$. Together with $\left\{  w_{1},M\left(
v_{1}\right)  ,v_{1},v_{2},M\left(  v_{2}\right)  ,w_{2}\right\}  $ this path
produces a cycle of odd length in contradiction with the hypothesis on $G$
being a bipartite graph.

To complete the proof of independence of the set
\[
X^{1}=X\cup M\left(  \left(  N\left(  X\right)  \cap S\right)  -M\left(
X\right)  \right)
\]
it is enough to demonstrate that there are no edges connecting vertices of $X$
and $M\left(  \left(  N\left(  X\right)  \cap S\right)  -M\left(  X\right)
\right)  $. \begin{figure}[h]
\setlength{\unitlength}{1.0cm} \begin{picture}(5,4)\thicklines
\put(6.7,1){\oval(10,1.5)}
\put(5,1){\oval(4,1)}
\put(5,3){\oval(4,1)}
\put(6.7,3){\oval(10,1.5)}
\put(9,1){\oval(2.5,1)}
\put(9,3){\oval(2.5,1)}
\multiput(3.5,1)(0.7,0){2}{\circle*{0.29}}
\multiput(3.5,3)(0.7,0){2}{\circle*{0.29}}
\multiput(5.5,1)(0,2){2}{\circle*{0.29}}
\multiput(6.5,1)(0,2){2}{\circle*{0.29}}
\multiput(4.5,1)(0.35,0){3}{\circle*{0.1}}
\multiput(4.5,3)(0.35,0){3}{\circle*{0.1}}
\put(5.5,1){\line(0,1){2}}
\put(5.5,1){\line(3,2){3}}
\put(6.5,1){\line(0,1){2}}
\put(6.5,1){\line(3,2){3}}
\put(6.5,1){\line(1,1){2}}
\multiput(3.5,1)(0.7,0){2}{\line(0,1){2}}
\put(1,1){\makebox(0,0){$V-S$}}
\put(0.2,2.2){\makebox(0,0){$G$}}
\put(2.5,1){\makebox(0,0){$X$}}
\put(2.45,3){\makebox(0,0){$M(X)$}}
\put(1.2,3){\makebox(0,0){$S$}}
\multiput(8.5,1)(1,0){2}{\circle*{0.29}}
\multiput(8.5,3)(1,0){2}{\circle*{0.29}}
\put(8.5,1){\line(0,1){2}}
\put(9.5,1){\line(0,1){2}}
\put(10.85,1){\makebox(0,0){$M(Y)$}}
\put(10.75,3){\makebox(0,0){$Y$}}
\end{picture}\caption{$S\in\Omega(G)$, $Y=\left(  N\left(  X\right)  \cap
S\right)  -M\left(  X\right)  ${ and }$X^{1}=X\cup M\left(  Y\right)  $.}%
\end{figure}
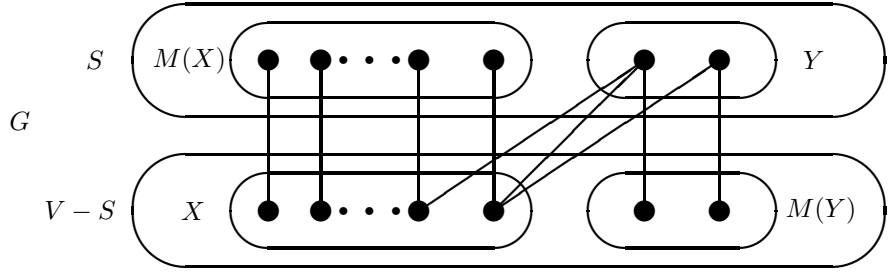

Assume, to the contrary, that there is $vw\in E$, such that $v\in M\left(
\left(  N\left(  X\right)  \cap S\right)  -M\left(  X\right)  \right)  $ and
$w\in X$. Since $M\left(  v\right)  \in\left(  N\left(  X\right)  \cap
S\right)  -M\left(  X\right)  $ and $G\left[  X\cup M\left(  X\right)
\right]  $ is connected, it follows that there exists a path with an odd
number of edges connecting $M\left(  v\right)  $ to $w$. This path together
with the edges $vw$ and $vM\left(  v\right)  $ produces cycle of odd length,
in contradiction with the bipartiteness of $G$.

Finally, since $G\left[  X\cup M\left(  X\right)  \right]  $ is connected,
$G\left[  X^{1}\cup M\left(  X^{1}\right)  \right]  $ is connected as well, by
definitions of set functions $N$ and $M$.
\end{proof}

Theorem \ref{th6} claims that $\mathrm{\ker}(G)\subseteq\mathrm{core}(G)$ for
every graph.

\begin{theorem}
\cite{LevMan2011b}\label{th10} If $G$ is a bipartite graph, then
$\mathrm{\ker}(G)=\mathrm{core}(G)$.
\end{theorem}

\begin{proof}
[Alternative Proof]The assertions are clearly true, whenever $\mathrm{core}%
(G)=\emptyset$, i.e., for $G$ having a perfect matching. Assume that
$\mathrm{core}(G)\neq\emptyset$.

Let $S\in\Omega\left(  G\right)  $ and $M$ be a maximum matching. By Theorem
\ref{Th5}\emph{(i)}, $M$ matches $V\left(  G\right)  -S$ into $S$, and
$N(\mathrm{core}(G))$ into $\mathrm{core}(G)$.

According to Theorem \ref{th9}\emph{(ii)}, it is sufficient to show that there
is no set $Z\subseteq N\left(  \mathrm{core}(G)\right)  $, $Z\neq\emptyset$,
such that $\left\vert N\left(  Z\right)  \cap\mathrm{core}(G)\right\vert
=\left\vert Z\right\vert $.

Suppose, to the contrary, that there exists a non-empty set $Z\subseteq$
$N\left(  \mathrm{core}(G)\right)  $ such that $\left\vert N\left(  Z\right)
\cap\mathrm{core}(G)\right\vert =\left\vert Z\right\vert $. Let $Z_{0}$ be a
minimal non-empty subset of $N\left(  \mathrm{core}(G)\right)  $ enjoying this equality.

Clearly, $H=G\left[  Z_{0}\cup M\left(  Z_{0}\right)  \right]  $ is bipartite,
because it is a subgraph of a bipartite graph. Moreover, the restriction of
$M$ on $H$ is a perfect matching.

\textit{Claim 1.} $Z_{0}$ is independent.

Since $H$ is a bipartite graph with a perfect matching it has two maximum
independent sets at least. Hence there exists $W\in\Omega\left(  H\right)  $
different from $M\left(  Z_{0}\right)  $. Thus $W\cap Z_{0}\neq\emptyset$.
Therefore, $N\left(  W\cap Z_{0}\right)  \cap\mathrm{core}(G)=M\left(  W\cap
Z_{0}\right)  $. Consequently,
\[
\left\vert N\left(  W\cap Z_{0}\right)  \cap\mathrm{core}(G)\right\vert
=\left\vert M\left(  W\cap Z_{0}\right)  \right\vert =\left\vert W\cap
Z_{0}\right\vert .
\]
Finally, $W\cap Z_{0}=Z_{0}$, because $Z_{0}$ has been chosen as a minimal
subset of $N\left(  \mathrm{core}(G)\right)  $ such that $\left\vert N\left(
Z_{0}\right)  \cap\mathrm{core}(G)\right\vert =\left\vert Z_{0}\right\vert $.
Since $\left\vert Z_{0}\right\vert =\alpha\left(  H\right)  =\left\vert
W\right\vert $ we conclude with $W=Z_{0}$, which means, in particular, that
$Z_{0}$ is independent.

\textit{Claim 2.} $H$ is a connected graph.

Otherwise, for any connected component of $H$, say $\tilde{H}$, the set
$V\left(  \tilde{H}\right)  \cap Z_{0}$ contradicts the minimality property of
$Z_{0}$.

\textit{Claim 3.} $Z_{0}\cup$ $\left(  \mathrm{core}(G)-M\left(  Z_{0}\right)
\right)  $ is independent.

By Claim 1 $Z_{0}$ is independent. The equality $\left\vert N\left(
Z_{0}\right)  \cap\mathrm{core}(G)\right\vert =\left\vert Z_{0}\right\vert $
implies $N\left(  Z_{0}\right)  \cap\mathrm{core}(G)=M\left(  Z_{0}\right)  $,
which means that there are no edges connecting $Z_{0}$ and $\mathrm{core}%
(G)-M\left(  Z_{0}\right)  $. Consequently, $Z_{0}\cup$ $\left(
\mathrm{core}(G)-M\left(  Z_{0}\right)  \right)  $ is independent.

\textit{Claim 4.} $Z_{0}\cup$ $\left(  \mathrm{core}(G)-M\left(  Z_{0}\right)
\right)  $ is included in a maximum independent set.

Let $Z_{i}=M\left(  \left(  N\left(  Z_{i-1}\right)  \cap S\right)  -M\left(
Z_{i-1}\right)  \right)  ,1\leq i<\infty$. By Lemma \ref{lem1} all the sets
$Z^{i}=\bigcup\limits_{0\leq j\leq i}Z_{j},1\leq i<\infty$ are independent.
Define%
\[
Z^{\infty}=\bigcup\limits_{0\leq i\leq\infty}Z_{i},
\]
which is, actually, the largest set in the sequence $\left\{  Z^{i},1\leq
i<\infty\right\}  $.\begin{figure}[h]
\setlength{\unitlength}{1.0cm} \begin{picture}(5,4)\thicklines
\put(7.5,1){\oval(11,1.5)}
\put(6.7,3){\oval(12.5,1.7)}
\put(3.5,3){\oval(4.35,1.2)}
\put(2.5,3){\oval(1.5,0.8)}
\put(4.5,1){\oval(2,0.8)}
\put(4.5,3){\oval(2,0.8)}
\put(7.5,1){\oval(2,0.8)}
\put(7.5,3){\oval(2,0.8)}
\put(10.5,1){\oval(2,0.8)}
\put(10.5,3){\oval(2,0.8)}
\multiput(4,1)(1,0){2}{\circle*{0.29}}
\multiput(4,3)(1,0){2}{\circle*{0.29}}
\multiput(4,1)(1,0){2}{\line(0,1){2}}
\put(4,1){\line(3,2){3}}
\put(5,1){\line(3,2){3}}
\put(5,1){\line(1,1){2}}
\multiput(7,1)(1,0){2}{\circle*{0.29}}
\multiput(7,3)(1,0){2}{\circle*{0.29}}
\put(7,1){\line(0,1){2}}
\put(8,1){\line(0,1){2}}
\put(7,1){\line(3,2){3}}
\put(8,1){\line(3,2){3}}
\put(8,1){\line(1,1){2}}
\multiput(10,1)(1,0){2}{\circle*{0.29}}
\multiput(10,3)(1,0){2}{\circle*{0.29}}
\put(10,1){\line(0,1){2}}
\put(11,1){\line(0,1){2}}
\multiput(11.9,1)(0.35,0){3}{\circle*{0.1}}
\multiput(11.9,3)(0.35,0){3}{\circle*{0.1}}
\put(0.1,3){\makebox(0,0){$S$}}
\put(1,1){\makebox(0,0){$V-S$}}
\put(0.2,2){\makebox(0,0){$G$}}
\put(0.9,3){\makebox(0,0){$core$}}
\put(3,1){\makebox(0,0){$Z_{0}$}}
\put(4.5,3.1){\makebox(0,0){$Y_{0}$}}
\put(2.5,3.1){\makebox(0,0){$Q$}}
\put(6,1){\makebox(0,0){$Z_{1}$}}
\put(7.5,3.1){\makebox(0,0){$Y_{1}$}}
\put(9,1){\makebox(0,0){$Z_{2}$}}
\put(10.5,3.1){\makebox(0,0){$Y_{2}$}}
\end{picture}\caption{$S\in\Omega(G)$, $Q=\mathrm{core}\left(  G\right)
-M\left(  Z_{0}\right)  $, $Y_{0}=$ $M\left(  Z_{0}\right)  $, $Y_{1}=\left(
N\left(  Z_{0}\right)  -M\left(  Z_{0}\right)  \right)  \cap S$, $Y_{2}=...$,
and $Z_{i}=M\left(  Y_{i}\right)  ,i=1,2,...$ {.}}%
\end{figure}
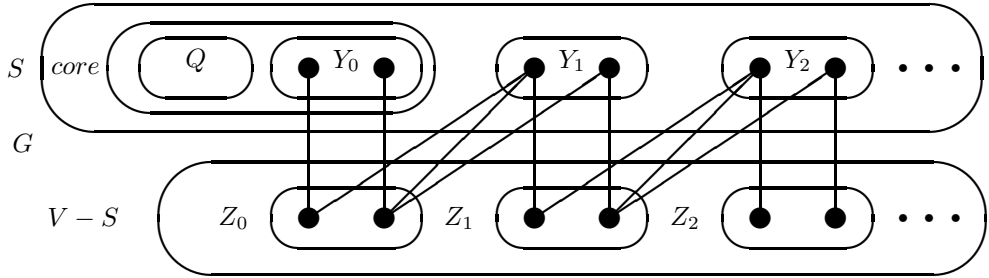

The inclusion
\[
Z_{0}\cup\left(  \mathrm{core}(G)-M\left(  Z_{0}\right)  \right)
\subseteq\left(  S-M\left(  Z^{\infty}\right)  \right)  \cup Z^{\infty}%
\]
is justified by the definition of $Z^{\infty}$.

Since $\left\vert M\left(  Z^{\infty}\right)  \right\vert =\left\vert
Z^{\infty}\right\vert $ we obtain $\left\vert \left(  S-M\left(  Z^{\infty
}\right)  \right)  \cup Z^{\infty}\right\vert =\left\vert S\right\vert $.
According to the definition of $Z^{\infty}$ the set
\[
\left(  N\left(  Z^{\infty}\right)  \cap S\right)  -M\left(  Z^{\infty
}\right)
\]
is empty. In other words, the set $\left(  S-M\left(  Z^{\infty}\right)
\right)  \cup Z^{\infty}$ is independent. Therefore, we arrive at
\[
\left(  S-M\left(  Z^{\infty}\right)  \right)  \cup Z^{\infty}\in\Omega\left(
G\right)  .
\]

Consequently, $\left(  S-M\left(  Z^{\infty}\right)  \right)  \cup Z^{\infty}$
is a desired enlargement of $Z_{0}\cup$ $\left(  \mathrm{core}(G)-M\left(
Z_{0}\right)  \right)  $.

\textit{Claim 5.} $\mathrm{core}(G)\cap\left(  \left(  S-M\left(  Z^{\infty
}\right)  \right)  \cup Z^{\infty}\right)  =\mathrm{core}(G)-M\left(
Z_{0}\right)  $.

The only part of $\left(  S-M\left(  Z^{\infty}\right)  \right)  \cup
Z^{\infty}$ that interacts with $\mathrm{core}(G)$ is the subset
\[
Z_{0}\cup\left(  \mathrm{core}(G)-M\left(  Z_{0}\right)  \right)  .
\]
Hence we obtain
\begin{gather*}
\mathrm{core}(G)\cap\left(  \left(  S-M\left(  Z^{\infty}\right)  \right)
\cup Z^{\infty}\right)  =\\
=\mathrm{core}(G)\cap\left(  Z_{0}\cup\left(  \mathrm{core}(G)-M\left(
Z_{0}\right)  \right)  \right)  =\mathrm{core}(G)-M\left(  Z_{0}\right)  .
\end{gather*}

Since $Z_{0}$ is non-empty, by Claim 5 we arrive at the following
contradiction
\[
\mathrm{core}(G)\nsubseteq\left(  S-M\left(  Z^{\infty}\right)  \right)  \cup
Z^{\infty}\in\Omega\left(  G\right)  .
\]

Finally, we conclude with the fact there is no set $Z\subseteq$ $N\left(
\mathrm{core}(G)\right)  ,Z\neq\emptyset$ such that $\left\vert N\left(
Z\right)  \cap\mathrm{core}(G)\right\vert =\left\vert Z\right\vert $, which,
by Theorem \ref{th9}, means that $\mathrm{core}(G)$ and $\mathrm{\ker}(G)$ coincide.
\end{proof}

Notice that there are non-bipartite graphs enjoying the equality
$\mathrm{\ker}(G)=\mathrm{core}(G)$; e.g., the graphs from Figure \ref{fig14},
where only $G_{1}$ is a K\"{o}nig-Egerv\'{a}ry graph. 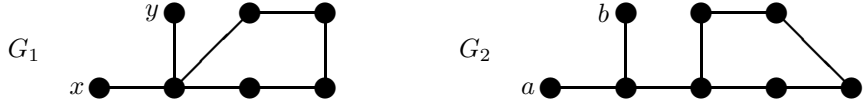
\begin{figure}[h]
\setlength{\unitlength}{1cm}\begin{picture}(5,1.2)\thicklines
\multiput(2,0)(1,0){4}{\circle*{0.29}}
\multiput(3,1)(1,0){3}{\circle*{0.29}}
\put(2,0){\line(1,0){3}}
\put(3,0){\line(0,1){1}}
\put(5,0){\line(0,1){1}}
\put(4,1){\line(1,0){1}}
\put(3,0){\line(1,1){1}}
\put(1.7,0){\makebox(0,0){$x$}}
\put(2.7,1){\makebox(0,0){$y$}}
\put(1,0.5){\makebox(0,0){$G_{1}$}}
\multiput(8,0)(1,0){5}{\circle*{0.29}}
\multiput(9,1)(1,0){3}{\circle*{0.29}}
\put(8,0){\line(1,0){4}}
\put(9,0){\line(0,1){1}}
\put(10,0){\line(0,1){1}}
\put(10,1){\line(1,0){1}}
\put(11,1){\line(1,-1){1}}
\put(7.7,0){\makebox(0,0){$a$}}
\put(8.7,1){\makebox(0,0){$b$}}
\put(7,0.5){\makebox(0,0){$G_{2}$}}
\end{picture}\caption{$\mathrm{core}(G_{1})=\ker\left(  G_{1}\right)
=\{x,y\}$ and $\mathrm{core}(G_{2})=\ker\left(  G_{2}\right)  =\{a,b\}$.}%
\label{fig14}%
\end{figure}

There is a non-bipartite K\"{o}nig-Egerv\'{a}ry graph $G$, such that
$\mathrm{\ker}(G)\neq\mathrm{core}(G)$. For instance, the graph $G_{1}$ from
Figure \ref{fig222} has $\mathrm{\ker}(G_{1})=\left\{  x,y\right\}  $, while
$\mathrm{core}(G_{1})=\left\{  x,y,u,v\right\}  $. The graph $G_{2}$ from
Figure \ref{fig222} has $\mathrm{\ker}(G_{2})=\emptyset$, while $\mathrm{core}%
(G_{2})=\left\{  w\right\}  $.

\begin{figure}[h]
\setlength{\unitlength}{1cm}\begin{picture}(5,1.3)\thicklines
\multiput(4,0)(1,0){4}{\circle*{0.29}}
\multiput(3,1)(1,0){5}{\circle*{0.29}}
\put(4,0){\line(1,0){3}}
\put(4,0){\line(0,1){1}}
\put(3,1){\line(1,-1){1}}
\put(5,0){\line(0,1){1}}
\put(5,0){\line(1,1){1}}
\put(5,1){\line(1,-1){1}}
\put(6,0){\line(0,1){1}}
\put(7,0){\line(0,1){1}}
\put(2.7,1){\makebox(0,0){$x$}}
\put(3.7,1){\makebox(0,0){$y$}}
\put(4.7,1){\makebox(0,0){$u$}}
\put(6.3,1){\makebox(0,0){$v$}}
\put(2,0.5){\makebox(0,0){$G_{1}$}}
\multiput(10,0)(1,0){2}{\circle*{0.29}}
\multiput(12,0)(0,1){2}{\circle*{0.29}}
\put(10,0){\line(1,0){2}}
\put(11,0){\line(1,1){1}}
\put(12,0){\line(0,1){1}}
\put(10,0.3){\makebox(0,0){$w$}}
\put(9,0.5){\makebox(0,0){$G_{2}$}}
\end{picture}\caption{Both $G_{1}$ and $G_{2}$\ are K\"{o}nig-Egerv\'{a}ry
graphs. Only $G_{2}$\ has a perfect matching.}%
\label{fig222}%
\end{figure}
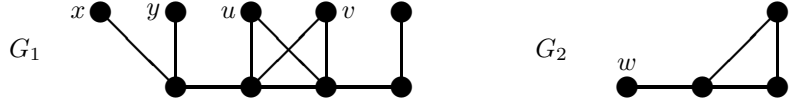

\section{$\mathrm{\ker}\left(  G\right)  $ and $\mathrm{diadem}(G)$ in
K\"{o}nig-Egerv\'{a}ry graphs}

There is a non-K\"{o}nig-Egerv\'{a}ry graph $G$ with $V\left(  G\right)
=N\left(  \mathrm{core}(G)\right)  \cup\mathrm{corona}(G)$; e.g., the graph
$G$ from Figure \ref{fig1777}.

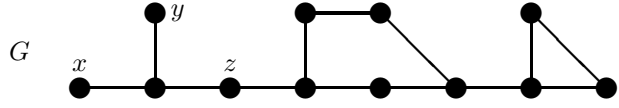
\begin{figure}[h]
\setlength{\unitlength}{1cm}\begin{picture}(5,1.2)\thicklines
\multiput(4,0)(1,0){8}{\circle*{0.29}}
\multiput(5,1)(2,0){2}{\circle*{0.29}}
\multiput(8,1)(2,0){2}{\circle*{0.29}}
\put(4,0){\line(1,0){7}}
\put(5,0){\line(0,1){1}}
\put(7,0){\line(0,1){1}}
\put(7,1){\line(1,0){1}}
\put(8,1){\line(1,-1){1}}
\put(10,0){\line(0,1){1}}
\put(10,1){\line(1,-1){1}}
\put(4,0.3){\makebox(0,0){$x$}}
\put(5.3,1){\makebox(0,0){$y$}}
\put(6,0.3){\makebox(0,0){$z$}}
\put(3.2,0.5){\makebox(0,0){$G$}}
\end{picture}
\caption{$G$ is not a K\"{o}nig-Egerv\'{a}ry graph, and $\mathrm{core}%
(G)=\left\{  x,y,z\right\}  $. }%
\label{fig1777}%
\end{figure}

\begin{theorem}
\label{th11}If $G$ is a K\"{o}nig-Egerv\'{a}ry graph, then

\emph{(i)}$\ \left\vert \mathrm{corona}(G)\right\vert +\left\vert
\mathrm{core}(G)\right\vert =2\alpha\left(  G\right)  $;

\emph{(ii) }$\mathrm{diadem}(G)=\mathrm{corona}(G)$, while $\mathrm{diadem}%
(G)\subseteq\mathrm{corona}(G)$ is true for every graph;

\emph{(iii) }$\left\vert \ker\left(  G\right)  \right\vert +\left\vert
\mathrm{diadem}\left(  G\right)  \right\vert \leq2\alpha\left(  G\right)  $.
\end{theorem}

\begin{proof}
\emph{(i) }Using Theorems \ref{Th5}\emph{(ii) }and \ref{th8}, we infer that
\begin{gather*}
\left\vert \mathrm{corona}(G)\right\vert +\left\vert \mathrm{core}%
(G)\right\vert =\left\vert \mathrm{corona}(G)\right\vert +\left\vert N\left(
\mathrm{core}(G)\right)  \right\vert +\left\vert \mathrm{core}(G)\right\vert
-\left\vert N\left(  \mathrm{core}(G)\right)  \right\vert =\\
=\left\vert V\left(  G\right)  \right\vert +d\left(  G\right)  =\alpha\left(
G\right)  +\mu\left(  G\right)  +d\left(  G\right)  =2\alpha\left(  G\right)
.
\end{gather*}
as claimed.

\emph{(ii)} Every $S\in\Omega\left(  G\right)  $ is a critical set, by Theorem
\ref{th5}. Hence we deduce that $\mathrm{corona}(G)\subseteq\mathrm{diadem}%
(G)$. On the other hand, for every graph each critical independent set is
included in a maximum independent set, according to Theorem \ref{th3}. Thus,
we infer that $\mathrm{diadem}(G)\subseteq\mathrm{corona}(G)$. Consequently,
the equality $\mathrm{diadem}(G)=\mathrm{corona}(G)$ holds.

\emph{(iii)} It follows by combining parts \emph{(i),(ii)} and Theorem
\ref{th6}.
\end{proof}

Notice that the graph from Figure \ref{fig1777} has $\left\vert
\mathrm{corona}(G)\right\vert +\left\vert \mathrm{core}(G)\right\vert
=13>12=2\alpha\left(  G\right)  $.

For a K\"{o}nig-Egerv\'{a}ry graph with $\left\vert \ker\left(  G\right)
\right\vert +\left\vert \text{\textrm{diadem}}\left(  G\right)  \right\vert
<2\alpha\left(  G\right)  $ see Figure \ref{fig222}. Figure \ref{fig1777}
shows that it is possible for a graph to have $\mathrm{diadem}%
(G)\varsubsetneqq\mathrm{corona}(G)$ and $\mathrm{\ker}(G)\varsubsetneqq
\mathrm{core}(G)$.

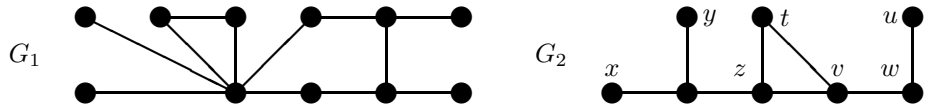
\begin{figure}[h]
\setlength{\unitlength}{1cm}\begin{picture}(5,1.2)\thicklines
\multiput(2,1)(1,0){6}{\circle*{0.29}}
\multiput(4,0)(1,0){4}{\circle*{0.29}}
\put(2,0){\circle*{0.29}}
\put(2,0){\line(1,0){5}}
\put(2,1){\line(2,-1){2}}
\put(3,1){\line(1,-1){1}}
\put(3,1){\line(1,0){1}}
\put(4,0){\line(0,1){1}}
\put(4,0){\line(1,1){1}}
\put(5,1){\line(1,0){2}}
\put(6,0){\line(0,1){1}}
\put(1.2,0.5){\makebox(0,0){$G_{1}$}}
\multiput(9,0)(1,0){5}{\circle*{0.29}}
\multiput(10,1)(1,0){2}{\circle*{0.29}}
\put(13,1){\circle*{0.29}}
\put(9,0){\line(1,0){4}}
\put(10,0){\line(0,1){1}}
\put(11,0){\line(0,1){1}}
\put(11,1){\line(1,-1){1}}
\put(13,0){\line(0,1){1}}
\put(9,0.3){\makebox(0,0){$x$}}
\put(10.3,1){\makebox(0,0){$y$}}
\put(11.3,1){\makebox(0,0){$t$}}
\put(12.7,1){\makebox(0,0){$u$}}
\put(10.7,0.3){\makebox(0,0){$z$}}
\put(12,0.3){\makebox(0,0){$v$}}
\put(12.7,0.3){\makebox(0,0){$w$}}
\put(8.2,0.5){\makebox(0,0){$G_{2}$}}
\end{picture}\caption{$G_{1}$ is a non-bipartite K\"{o}nig-Egerv\'{a}ry graph,
such that $\mathrm{\ker}(G_{1})=\mathrm{core}(G_{1})$ and \textrm{diadem}%
$\left(  G_{1}\right)  =\mathrm{corona}(G_{1})$; $G_{2}$ is a
non-K\"{o}nig-Egerv\'{a}ry graph, such that $\mathrm{\ker}(G_{2}%
)=\mathrm{core}(G_{2})=\{x,y\}$; \textrm{diadem}$\left(  G_{2}\right)
\cup\{z,t,v,w\}=\mathrm{corona}(G_{2})$.}%
\label{fig17888}%
\end{figure}The combination of $\mathrm{diadem}(G)\varsubsetneqq
\mathrm{corona}(G)$ and $\mathrm{\ker}(G)=\mathrm{core}(G)$ is realized in
Figure \ref{fig17888}.

\begin{proposition}
\label{Cor1} Let $G=\left(  A,B,E\right)  $ be a bipartite graph.

\emph{(i)} \cite{Ore62} If $X=\ker_{A}\left(  G\right)  $ \textit{and }%
$Y$\textit{ is a }$B$-\textit{critical set, then }$X\cap N\left(  Y\right)
=N\left(  X\right)  \cap Y=\emptyset$;

\emph{(ii)} \cite{Ore55} $\ker_{A}\left(  G\right)  \cap N\left(  \ker
_{B}\left(  G\right)  \right)  =N\left(  \ker_{A}\left(  G\right)  \right)
\cap\ker_{B}\left(  G\right)  =\emptyset$.
\end{proposition}

Now we are ready to describe both $\ker$ and \textrm{diadem} of a bipartite
graph in terms of its bipartition.

\begin{theorem}
Let $G=\left(  A,B,E\right)  $ be a bipartite graph. Then the following
assertions are true:

\emph{(i) }$\ker_{A}\left(  G\right)  \cup$ $\ker_{B}\left(  G\right)
=\ker\left(  G\right)  $;

\emph{(ii) }$\left\vert \ker\left(  G\right)  \right\vert +\left\vert
\mathrm{diadem}\left(  G\right)  \right\vert =2\alpha\left(  G\right)  $;

\emph{(iii) }$\left\vert \ker_{A}\left(  G\right)  \right\vert +\left\vert
\mathrm{diadem}_{B}\left(  G\right)  \right\vert =\left\vert \ker_{B}\left(
G\right)  \right\vert +\left\vert \mathrm{diadem}_{A}\left(  G\right)
\right\vert =\alpha\left(  G\right)  $;

\emph{(iv) }$\mathrm{diadem}_{A}\left(  G\right)  \cup$ $\mathrm{diadem}%
_{B}\left(  G\right)  =$ $\mathrm{diadem}\left(  G\right)  $.
\end{theorem}

\begin{proof}
\emph{(i)} By Theorem \ref{Th4}\emph{(iii)}, $\ker_{A}\left(  G\right)  \cup$
$\ker_{B}\left(  G\right)  $ is critical in $G$. Moreover, the set $\ker
_{A}\left(  G\right)  \cup$ $\ker_{B}\left(  G\right)  $ is independent in
accordance with Proposition \ref{Cor1}\emph{(ii)}. Assume that $\ker
_{A}\left(  G\right)  \cup$ $\ker_{B}\left(  G\right)  $ is not minimal. Hence
the unique minimal $d$-critical set of $G$, say $Z$, is a proper subset of
$\ker_{A}\left(  G\right)  \cup$ $\ker_{B}\left(  G\right)  $, by Theorem
\ref{th4}\emph{(iii)}. According to Theorem \ref{Th4}\emph{(iv)}, $Z_{A}=Z\cap
A$ is an $A$-critical set, which implies $\ker_{A}\left(  G\right)  \subseteq
Z_{A}$, and similarly, $\ker_{B}\left(  G\right)  \subseteq Z_{B}$.
Consequently, we get that $\ker_{A}\left(  G\right)  \cup$ $\ker_{B}\left(
G\right)  \subseteq Z$, in contradiction with the fact that $\ker_{A}\left(
G\right)  \cup$ $\ker_{B}\left(  G\right)  \neq Z\subset\ker_{A}\left(
G\right)  \cup$ $\ker_{B}\left(  G\right)  $.

\emph{(ii), (iii), (iv) }By Proposition \ref{Cor1}\emph{(i)}, we have
\[
\left\vert \ker_{A}\left(  G\right)  \right\vert -\delta_{0}(A)+\left\vert
\text{\textrm{diadem}}_{B}\left(  G\right)  \right\vert =\left\vert N\left(
\ker_{A}\left(  G\right)  \right)  \right\vert +\left\vert
\text{\textrm{diadem}}_{B}\left(  G\right)  \right\vert \leq\left\vert
B\right\vert .
\]
Hence, according to Theorem \ref{Th4}\emph{(ii)}, it follows that
\[
\left\vert \ker_{A}\left(  G\right)  \right\vert +\left\vert
\text{\textrm{diadem}}_{B}\left(  G\right)  \right\vert \leq\left\vert
B\right\vert +\delta_{0}(A)=\alpha\left(  G\right)  .
\]
Changing the roles of $A$ and $B$, we obtain
\[
\left\vert \ker_{B}\left(  G\right)  \right\vert +\left\vert
\text{\textrm{diadem}}_{A}\left(  G\right)  \right\vert \leq\alpha\left(
G\right)  .
\]

By Theorem \ref{Th4}\emph{(iv)}, \textrm{diadem}$\left(  G\right)  \cap A$ is
$A$-critical and \textrm{diadem}$\left(  G\right)  \cap B$ is $B$-critical.
Hence \textrm{diadem}$\left(  G\right)  \cap A\subseteq$ \textrm{diadem}%
$_{A}\left(  G\right)  $ and \textrm{diadem}$\left(  G\right)  \cap
B\subseteq$ \textrm{diadem}$_{B}\left(  G\right)  $. It implies both the
inclusion $\mathrm{diadem}\left(  G\right)  \subseteq\mathrm{diadem}%
_{A}\left(  G\right)  \cup\mathrm{diadem}_{B}\left(  G\right)  $, and the
inequality
\[
\left\vert \mathrm{diadem}\left(  G\right)  \right\vert \leq\left\vert
\mathrm{diadem}_{A}\left(  G\right)  \right\vert +\left\vert \mathrm{diadem}%
_{B}\left(  G\right)  \right\vert .
\]

Combining Theorem \ref{th10}, Theorem \ref{th11}\emph{(i),(ii)}, and part
\emph{(i)} with the above inequalities, we deduce%
\begin{gather*}
2\alpha\left(  G\right)  \geq\left\vert \ker_{A}\left(  G\right)  \right\vert
+\left\vert \ker_{B}\left(  G\right)  \right\vert +\left\vert
\text{\textrm{diadem}}_{A}\left(  G\right)  \right\vert +\left\vert
\text{\textrm{diadem}}_{B}\left(  G\right)  \right\vert \geq\\
\geq\left\vert \ker\left(  G\right)  \right\vert +\left\vert
\text{\textrm{diadem}}\left(  G\right)  \right\vert =\left\vert
\text{\textrm{core}}\left(  G\right)  \right\vert +\left\vert
\text{\textrm{corona}}\left(  G\right)  \right\vert =2\alpha\left(  G\right)
.
\end{gather*}

Consequently, we infer that%
\begin{gather*}
\left\vert \text{\textrm{diadem}}_{A}\left(  G\right)  \right\vert +\left\vert
\text{\textrm{diadem}}_{B}\left(  G\right)  \right\vert =\left\vert
\text{\textrm{diadem}}\left(  G\right)  \right\vert ,\\
\left\vert \ker\left(  G\right)  \right\vert +\left\vert \text{\textrm{diadem}%
}\left(  G\right)  \right\vert =2\alpha\left(  G\right)  ,\\
\left\vert \ker_{A}\left(  G\right)  \right\vert +\left\vert
\text{\textrm{diadem}}_{B}\left(  G\right)  \right\vert =\left\vert \ker
_{B}\left(  G\right)  \right\vert +\left\vert \text{\textrm{diadem}}%
_{A}\left(  G\right)  \right\vert =\alpha\left(  G\right)  .
\end{gather*}
Since $\mathrm{diadem}\left(  G\right)  \subseteq\mathrm{diadem}_{A}\left(
G\right)  \cup\mathrm{diadem}_{B}\left(  G\right)  $ and $\mathrm{diadem}%
_{A}\left(  G\right)  \cap\mathrm{diadem}_{B}\left(  G\right)  =\emptyset$, we
finally obtain that
\[
\mathrm{diadem}_{A}\left(  G\right)  \cup\mathrm{diadem}_{B}\left(  G\right)
=\mathrm{diadem}\left(  G\right)  ,
\]
as claimed.
\end{proof}

\section{Conclusions}

In this paper we focus on interconnections between $\ker$, \textrm{core,
diadem,} and \textrm{corona}. In \cite{LevManLemma2011} we showed that
$2\alpha\left(  G\right)  \leq\left\vert \text{\textrm{core}}\left(  G\right)
\right\vert +\left\vert \text{\textrm{corona}}\left(  G\right)  \right\vert $
is true for every graph, while the equality holds whenever $G$ is a
K\"{o}nig-Egerv\'{a}ry graph, by Theorem \ref{th11}\emph{(i)}.

According to Theorem \ref{th6}, $\mathrm{\ker}(G)\subseteq\mathrm{core}(G)$
for every graph. On the other hand, Theorem \ref{th3}\emph{ }implies the
inclusion \textrm{diadem}$\left(  G\right)  \subseteq\mathrm{corona}(G)$. Hence%

\[
\left\vert \ker\left(  G\right)  \right\vert +\left\vert \text{\textrm{diadem}%
}\left(  G\right)  \right\vert \leq\left\vert \text{\textrm{core}}\left(
G\right)  \right\vert +\left\vert \text{\textrm{corona}}\left(  G\right)
\right\vert
\]
for each graph $G$. These remarks together with Theorem \ref{th11}\emph{(iii)}
motivate the following.

\begin{conjecture}
$\left\vert \ker\left(  G\right)  \right\vert +\left\vert
\text{\textrm{diadem}}\left(  G\right)  \right\vert \leq2\alpha\left(
G\right)  $ is true for every graph $G$.
\end{conjecture}

When it is proved one can conclude that the following inequalities:
\[
\left\vert \ker\left(  G\right)  \right\vert +\left\vert \text{\textrm{diadem}%
}\left(  G\right)  \right\vert \leq2\alpha\left(  G\right)  \leq\left\vert
\text{\textrm{core}}\left(  G\right)  \right\vert +\left\vert
\text{\textrm{corona}}\left(  G\right)  \right\vert
\]
hold for every graph $G$.

By Corollary \ref{cor2}, $\mathrm{core}(G)$ is critical for every
K\"{o}nig-Egerv\'{a}ry graph. It justifies the following.

\begin{problem}
Characterize graphs such that $\mathrm{core}(G)$ is a critical set.
\end{problem}

Theorem \ref{th10} claims that the sets $\mathrm{\ker}(G)$ and $\mathrm{core}%
(G)$ coincide for bipartite graphs. On the other hand, there are examples
showing that this equality holds even for some non-K\"{o}nig-Egerv\'{a}ry
graphs (see Figure \ref{fig14}). We propose the following.

\begin{problem}
Characterize graphs with $\ker\left(  G\right)  =\mathrm{core}(G)$.
\end{problem}

\section{Acknowledgments}

The authors would like to thank the organizers of the International Conference
in Discrete Mathematics (ICDM 2013) for an opportunity to give a special
invited talk including their recent findings.

\end{document}